\documentclass[conference]{IEEEtran}
\usepackage{etex}
\ifCLASSINFOpdf
\else
\fi
\usepackage{times,amsmath,epsfig,latexsym,multirow,array,amssymb,graphicx,multicol,stfloats,xcolor,textcomp,mathrsfs}

\usepackage{tikz,pgfplots}

\usepackage{pstricks}
\usepackage{pstricks-add}

\usepackage{caption2}

\newcounter{clemma}
\newtheorem{lemma}[clemma]{Lemma}
\newcounter{ctheorem}
\newtheorem{theorem}[ctheorem]{Theorem}

\usepackage{auto-pst-pdf}

\usepackage{algorithm}
\usepackage{algpseudocode}
\usepackage{balance}

\begin{document}
\title{CSIT Sharing over Finite Capacity Backhaul for Spatial Interference Alignment}


\author{\IEEEauthorblockN{Mohsen Rezaee$^\dagger$, Maxime Guillaud$^\dagger$, Fredrik Lindqvist$^\ddagger$}
\IEEEauthorblockA{
%Gu\ss hausstra\ss e 25/389\\
$^\dagger$Institute of
Telecommunications, Vienna University of Technology, \\Vienna, Austria -- e-mail: {{\small \texttt{\{mohsen.rezaee,guillaud\}@tuwien.ac.at}}}} 
%\end{minipage}
%\begin{minipage}{5cm}
$^\ddagger$Ericsson Research, Ericsson AB, Sweden -- e-mail: {{\small \texttt{fredrik.lindqvist@ericsson.com}}}
%\end{minipage}
}

%
%\author{\IEEEauthorblockN{Mohsen Rezaee, Maxime Guillaud, Fredrik Lindquist }
%\IEEEauthorblockA{
%%Gu\ss hausstra\ss e 25/389\\
%Institute of
%Telecommunications, Vienna University of Technology, \\Vienna, Austria -- e-mail: {{\small \texttt{\{guillaud,mohsen.rezaee\}@tuwien.ac.at}}}} 
%%\end{minipage}
%%\begin{minipage}{5cm}
%%\end{minipage}
%}
%


%Telephone: (800) 555--1212, Fax: (888) 555--1212}


\maketitle

\begin{abstract}
Cellular systems that employ time division duplexing (TDD) transmission are good candidates for implementation of interference alignment (IA) in the downlink since channel reciprocity enables the estimation of the channel state by the base stations (BS) in the uplink phase. However, the interfering BSs need to share their channel estimates via backhaul links of finite capacity. A quantization scheme is proposed which reduces the amount of information exchange (compared to conventional methods) required to achieve IA in such a system. The scaling (with the transmit power) of the number of bits to be exchanged between the BSs that is sufficient to preserve the multiplexing gain of IA is derived. \footnote{This work was supported by the FP7 project HIATUS (grant 265578) of the European Commission and by the Austrian Science Fund (FWF) through grant NFN SISE (S106).}
\end{abstract}


% \fi
\IEEEpeerreviewmaketitle

\section{Introduction}
\label{sec:intro}

Interference alignment (IA) is known to achieve the optimal degree of freedom (DoF) in interference channel (IC). This implies that at high signal to noise ratio (SNR) regime, IA improves the system throughput compared to the conventional orthogonal medium-sharing methods. However, implementation of IA in existing systems faces a lot of challenges. The necessity of channel state information (CSI) at the transmitters is one of the major issues which is not practical in many situations. Moreover, the accuracy of the CSI provided to the transmitters should increase as the power increases in order to guarantee the DoF gains promised by IA \cite{bolski}. Therefore transmission systems which acquire the CSI through feedback (such as frequency division duplex (FDD) systems) become less favorable for implementation of IA since the potential gains only appear at high powers.

For time division duplex (TDD) systems, every base station can estimate its downlink channels from the uplink transmission phase thanks to reciprocity. However, this local CSI is not sufficient, and the BSs need to share their channel estimates which can be carried out through backhaul links between BSs. These backhaul links generally have limited capacity, which should be exploited efficiently.
%Interference alignment is used in \cite{Kuser:CJ} to find the optimal DoF of the $K$ user interference channel. %IA has been introduced for the $K$-user MIMO IC in \cite{Gou_Jafar_DoF_MIMO_Kuser_IC_IT2010}.  %When the precoders are designed based on the quantized channels, the mismatch will result in an interference term which scales with the transmit power.

In scenarios where the receivers quantize and feed the CSI back to the transmitters, the problem is explored over frequency selective channels for single-antenna users in \cite{bolski} and for multiple-antenna users in \cite{rajac}. Both references provide DoF-achieving quantization schemes and establish the required scaling of the number of feedback bits. For alignment using spatial dimensions, \cite{itw} provides the scaling of feedback bits to achieve IA in MIMO IC. For the broadcast channel, the scaling of the feedback bits was characterized in \cite{Jindal}. % in order to attain the multiplexing gain which is achievable using perfect CSI. 
In \cite{Wiro}, quantization of the precoding matrix using random vector quantization (RVQ) codebooks is investigated which provides insights on the asymptotic optimality of RVQ. 
From another point of view, \cite{ImPerCh:Guill} provides an analysis of the effect of imperfect CSI on the mutual information of the interference alignment scheme. 
%The authors in \cite{korean} proposed a method to %find an extra unitary filter at each user that 
%reduce the quantization error w.r.t. the classical scheme; the method involves a computationally heavy iterative algorithm which must be ran for each codeword and for each channel realization.        
\\
\indent In this paper, we focus on the scenario where the BSs have perfect but local CSI, and must share it to achieve IA. A CSIT sharing scheme is proposed which reduces the amount of information exchange required for interference alignment in such a system. The scaling (with the transmit power) of the number of bits to be transferred which is sufficient to preserve the multiplexing gain that can be achieved using perfect CSI is derived. Moreover, a heuristic method is proposed to demonstrate the achievability of the DoF by simulations.

%The remainder of the paper is organized as follows. In Section \ref{sec:model}, the system model is described. The reformulation of the interference alignment problem is provided in Section \ref{sec:IA}. The limited feedback method  is presented in Section \ref{sec:LF}. Simulation results are presented in Section \ref{sec:simulation} and conclusions are drawn in Section \ref{sec:conclusion}.\\
{\it Notation:} Boldface lowercase and uppercase letters indicate vectors and matrices, respectively. ${\bf I}_N$ is the $N\times N$ identity matrix. The trace, conjugate, Hermitian transpose of a matrix or vector are denoted by ${\rm tr}(\cdot),(\cdot)^*, (\cdot)^{\rm H}$ respectively. The expectation operator is represented by ${\rm E(\cdot)}$. The Frobenius norm and the determinant of a matrix are denoted by $||\cdot ||_{\rm F}$ and $|\cdot |$ respectively. The maximum eigenvalue of a matrix is represented by $\lambda_{\rm max}(\cdot)$. A diagonal (resp. block diagonal) matrix is denoted by ${\rm diag(\cdot)}$ (resp. ${\rm Bdiag}(\cdot)$) with the argument elements (resp. blocks) on its diagonal. $\mathcal{N}(0,1)$ (resp. $\mathcal{CN}(0,1)$) denotes the real (resp. circularly symmetric complex) Gaussian distribution with zero mean and unit variance.
\section{System Model}
\label{sec:model}
An interference channel is considered in which $K$ base stations (BS) and $K$ users (one user in each cell) are considered as transmitters and receivers, respectively.  
For the sake of simplicity of the exposition, we focus on the symmetric case, and assume that each BS has $M$ antennas while each user is equipped with $N$ antennas. These results trivially generalize to non-homogeneous antenna numbers and per-user DoF as long as IA is feasible for the chosen problem dimensions.
Each BS employs a linear precoder to transmit $d$ data streams to its user. 
The received signal at user $i$ is denoted by
\begin{equation}\label{E1}
	{{\bf{y}}_i} = {{{{\bf H}}}_{ii}}{{{\bf {\bf V}}}_i}{{\bf{x}}_i} + \sum_{j=1, j \ne i}^K{{{{{\bf H}}}_{ij}}{{{{\bf {\bf V}}}}_j}{{\bf{x}}_j}}  + {{\bf{n}}_i}
\end{equation}
in which ${{{{\bf H}}}_{ij}} \in {\mathbb{C}^{N \times M}}$ is the channel matrix between BS $j$ and user $i$, ${{{{\bf {\bf V}}}}_{j}} \in {\mathbb{C}^{M \times d}}$ and ${{\bf{x}}_{j}} \in {\mathbb{C}^{d }}$ are the precoding matrix and the data vector of BS $j$, respectively. Furthermore, ${{\bf{n}}_{i}}$ is the additive noise at user $i$ whose entries are distributed according to $\mathcal{CN}(0,1)$. Assuming ${\rm {E}}\left(  {{\bf{x}}_j}{\bf{x}}_j^{\rm H} \right) = {\frac{P}{d}}{\bf I}_{d},{\rm{ }}\,\,\,j = 1, \ldots ,K$ and using truncated unitary precoders, the transmit power for each BS is equal to $P$. We further assume that the elements of the data symbol are i.i.d. Gaussian random variables. The channels are assumed to be generic \cite{yetisIA}; in particular, this includes channels with entries drawn independently from a continuous distribution.\\

\section{CSIT sharing for IA}
\label{sec:IA}

Let us consider TDD transmission, which enables the BSs to estimate their channels toward different users by exploiting the reciprocity of the wireless channel. Specifically, we assume that the $j$th BS estimates the channel matrices ${\bf H}_{ij}, \,\,i=1,\ldots,K,\, i\neq j$ (denoted by \emph{local CSI}) from the uplink phase, via reciprocity. We first assume that local CSI is known perfectly at BS $j$. However, global CSI (excluding the direct channels ${\bf H}_{ii}$) is required in order to design IA precoders. In this section we consider the topology of CSI exchange in the network, and work under the assumption that perfect local CSI is conveyed from each BS to a processing node which computes all precoders and provides them to the BSs.

%Assuming that global CSIT is available at a given location, the precoders ${{{\bf {\bf V}}}_i}$, $i=1\ldots K$ can be designed to align the interference at each user into a $N-d$ dimensional space % provided that the number of streams satisfies the feasibility condition, i.e., $d \leq \frac{M+N}{K+1}$. 
%Interference-free transmission for $d$ streams per user becomes possible for a feasible setting. We will further assume that $(K-1)N \geq M$, which represents the cases of interest where the BSs cannot zero force their interference for all the users simultaneously in the absence of alignment.  
 %It should be noted that all the results in this paper can be extended to different antenna configuration and different number of data streams per user as long as the feasibility condition is satisfied. 
%A solution to the IA problem exists (see \cite{Yetis_Gou_Jafar_Kayran_feasibility_of_IA09} and more recently \cite{Razaviyayn_etal_DoF_MIMO_IA_2011,Bresler_Cartwright_Tse_settling_feasibility_of_IA_symmetric_square_2011} for feasibility criteria  -- here we will assume that the dimensions and the considered channel realizations are such that the problem is feasible) iff there exist full rank precoding matrices ${ {\bf V}}_j ,\, j=1,...,K$ and projection matrices ${\bf U}_i \in {\mathbb{C}^{N \times d }}, \,i=1,...,K$ such that
Here we assume a feasible IA setting \cite{Gou_Jafar_DoF_MIMO_Kuser_IC_IT2010}, i.e. there exist precoding matrices ${ {\bf V}}_j ,\, j=1,...,K$ and projection matrices ${\bf U}_i \in {\mathbb{C}^{N \times d }}, \,i=1,...,K$ such that
\begin{eqnarray}\label{inoe}
&{\bf U}_i^{\rm H}{\bf H}_{ij}{\bf {\bf V}}_j={\bf 0} \,\,\, \,\  \,\,\,\,\, \forall i,j \in \{1,...,K\}, \,\, j \neq i,  
\\
&{\rm rank}({\bf U}_i^{\rm H}{\bf H}_{ii}{\bf V}_i)=d.  
\end{eqnarray}
Condition \eqref{inoe} can be rewritten as
\begin{equation}\begin{split}
{\bf U}_{-j}^{\rm H}{\bf H}_{j}{\bf V}_j={\bf 0}  \,\,\, \,\,\, \forall j \in \{1,...,K\},
\end{split}\end{equation}
in which ${\bf U}_{-j}={\rm Bdiag}({\bf U}_1,\ldots,{\bf U}_{j-1},{\bf U}_{j+1},\ldots,{\bf U}_K)$ and ${\bf H}_j=[{\bf H}_{1,j}^{\rm H},...,{\bf H}_{j-1,j}^{\rm H},{\bf H}_{j+1,j}^{\rm H},...,{\bf H}_{K,j}^{\rm H}]^{\rm H}$  is a $(K-1)N \times M$ matrix.

We will further assume that $(K-1)N > M$, which represents the cases where transmitter-side zero-forcing is not enough to eliminate all interference, and therefore IA is required. The following lemma highlights the intuition behind our CSI sharing scheme. 
\begin{lemma}\label{lemma_Grassmannian_feedback}
In order to design IA precoders, it is sufficient that each BS $j$ sends a point on the Grassmann manifold ${\mathcal G}_{(K-1)N,M}$ representing the column space of ${\bf H}_j$ to the IA processing node. 
\end{lemma}
\begin{proof}
Let ${\bf F}_j$ denote a $(K-1)N \times M$ matrix containing an orthonormal basis of the column space of ${\bf H}_j$, i.e.  ${\bf H}_j={\bf F}_j{\bf C}_j$ for some ${\bf C}_j$ (invertible almost surely for generic channels). According to our assumption that only the column space of ${\bf H}_i$ is known at the central unit, we can assume that the central unit has only access to a rotated version of ${\bf F}_j$, i.e., ${\bf F}_j{{\bf O}}_j$ for some unknown unitary matrix ${{\bf O}}_j$. 
We now show that alignment can be achieved based on the knowledge of ${{\bf F}_j{\bf O}_j}$ rather than of ${\bf H}_j$. Let us assume that the processing node designs a set $(\{{\tilde {\bf U}}_j\}_{j=1}^K,\{{\tilde {\bf V}}_j\}_{j=1}^K)$ of IA transmit precoders and receive projection filters for the channels $\{{{\bf F}_j{ {\bf O}}_j}\}_{j=1}^K$. Then,
\begin{eqnarray}
{\tilde {\bf U}}_{-j}^{\rm H}({\bf F}_j{ {\bf O}}_j){\tilde {\bf V}}_j={\bf 0}  & \!\! \Rightarrow & \!\! {\tilde {\bf U}}_{-j}^{\rm H}{\bf F}_j{\bf C}_j{\bf C}_j^{-1}{ {\bf O}}_j{\tilde {\bf V}}_j={\bf 0} \\
 & \!\! \Rightarrow & \!\! {\tilde {\bf U}}_{-j}^{\rm H}{\bf H}_{j}{\bf C}_j^{-1}{ {\bf O}}_j{\tilde {\bf V}}_j={\bf 0}.
\end{eqnarray}
This indicates that IA is achieved over the real channel by using ${\bf C}_j^{-1}{ {\bf O}}_j{\tilde {\bf V}}_j$ as precoder and ${\tilde {\bf U}}_j$ as the projection filter at user $j$. % instead of ${\bf U}_i^{\rm H}$. \\
Assuming that ${\tilde {\bf V}}_j$ is transmitted from the processing node back to BS $j$, and that ${{\bf O}}_j$ is known at BS $j$ since the reconstruction codebook of the processing node is known, the BS is in a position to compute the precoder ${\bf C}_j^{-1}{ {\bf O}}_j{\tilde {\bf V}}_j$.
\end{proof}
Note that the feedback of ${\tilde {\bf V}}_j$ from the processing node to BS $j$ also takes the form of a point on ${\mathcal G}_{M,d}$, and will be analyzed in further detail in the sequel.\\

\section{CSIT sharing over finite capacity links}
\label{sec:LF}

In this section, using the Grassmannian representation outlined in the previous section, we explore several scenarios where CSI is quantized and exchanged between the nodes over finite capacity links.
Three different scenarios regarding the CSIT sharing problem can be considered:\\
\indent I. The IA processing node is a separate central node that computes and distributes the IA precoders to the $K$ BSs,\\
\indent II. One BS also acts as the IA processing node, \\
\indent III. Each BS receives all the required CSI and independently computes the IA precoders.

In scenario I (Fig. \ref{fig:pic}(a)), the CSI (in the form of ${\bf F}_j$) is quantized yielding ${\hat {\bf F}}_j$ and sent to the central node. The central node computes the precoders and provides BS $j$ a quantized version ${\hat {\bf V}}_j$ of ${\tilde {\bf V}}_j$. Here we assume that each BS uses $N_b$ bits to quantize ${\bf F}_j$ and the central node uses $N_c$ bits to quantize ${\tilde {\bf V}}_j$. Therefore, the total number of bits exchanged over the network for scenario~I is equal to $K(N_b+N_c)$.  Scenario II can be considered as a particular example of scenario~I where one (bi-directional) BS-central node link is saved; the number of bits to be transferred in the network is $(K-1)(N_b+N_c)$. In scenario~III (Fig. \ref{fig:pic}(b)), the IA solution is computed independently at each BS, requiring global CSI at each BS. Therefore each BS needs to quantize and send its local CSI to all other $K-1$ BSs. The precoders are designed at the BSs and no further information 
 exchange is required. For simplicity of the exposition, we focus on scenario~I and characterize the scaling of $N_b$ and $N_c$ with $P$, noting that a generalization of the analysis to scenarios II and III is straightforward.

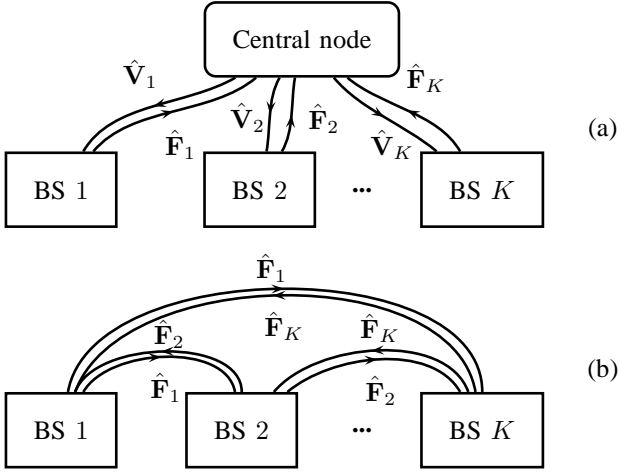
\begin{figure}
\psset{unit=0.80cm}
\psset{gridcolor=green, subgridcolor=yellow}
\begin{centering}
\begin{pspicture}(9,8)
%\psgrid
%\psset{xunit=1cm,yunit=1cm,runit=1cm}
\psset{linewidth=1pt}

\rput(1,5){\rnode{cn1}{\psframebox[fillstyle=solid, fillcolor=white, framesep=10pt]{BS $1$}}}
\rput(4.3,5){\rnode{cn2}{\psframebox[fillstyle=solid, fillcolor=white, framesep=10pt]{BS $2$}}} 
\rput(6,5){\text {\bf ...}}
\rput(8,5){\rnode{cnK}{\psframebox[fillstyle=solid, fillcolor=white, framesep=10pt]{BS $K$}}}

\rput(5,7.5){\rnode{cnB}{\psframebox[fillstyle=solid, fillcolor=white,framearc=.3, framesep=10pt]{Central node}}}

\nccurve[ArrowInside=->,ArrowInsidePos=0.5,angleA=50,angleB=220]{cn1}{cnB}
\bput(0.5){${\hat {\bf F}}_1$}

\nccurve[ArrowInside=->,ArrowInsidePos=0.5,angleA=60,angleB=-100]{cn2}{cnB}
\bput(0.5){${\hat {\bf F}}_2$}

\nccurve[ArrowInside=->,ArrowInsidePos=0.5,angleA=120,angleB=-40]{cnK}{cnB}
\bput(0.5){${\hat {\bf F}}_K$}

\nccurve[ArrowInside=->,ArrowInsidePos=0.5,angleA=210,angleB=60]{cnB}{cn1}
\bput(0.5){${\hat {\bf V}}_1$}

\nccurve[ArrowInside=->,ArrowInsidePos=0.5,angleA=-120,angleB=80]{cnB}{cn2}
\rput(4.1,6.2){${\hat {\bf V}}_2$}

\nccurve[ArrowInside=->,ArrowInsidePos=0.5,angleA=-50,angleB=140]{cnB}{cnK}
\rput(6.5,5.75){${\hat {\bf V}}_K$}

\rput(10,6){\text {(a)}}
\rput(10,2){\text {(b)}}

% noncentral :

\rput(1,1){\rnode{cn1}{\psframebox[fillstyle=solid, fillcolor=white, framesep=10pt]{BS $1$}}}
\rput(4,1){\rnode{cn2}{\psframebox[fillstyle=solid, fillcolor=white, framesep=10pt]{BS $2$}}} 
\rput(6,1){\text {\bf ...}}
\rput(8,1){\rnode{cnK}{\psframebox[fillstyle=solid, fillcolor=white, framesep=10pt]{BS $K$}}}

\nccurve[ArrowInside=->,ArrowInsidePos=0.5,angleA=60,angleB=100]{cn1}{cn2}
\bput(0.5){${\hat {\bf F}}_1$}

\nccurve[ArrowInside=->,ArrowInsidePos=0.5,angleA=40,angleB=120]{cn2}{cnK}
\bput(0.5){${\hat {\bf F}}_2$}

\nccurve[ArrowInside=->,ArrowInsidePos=0.5,angleA=100,angleB=70]{cnK}{cn1}
\aput(0.5){${\hat {\bf F}}_K$}

\nccurve[ArrowInside=->,ArrowInsidePos=0.5,angleA=90,angleB=70]{cn2}{cn1}
\rput(2.8,2.6){${\hat {\bf F}}_2$}

\nccurve[ArrowInside=->,ArrowInsidePos=0.5,angleA=110,angleB=50]{cnK}{cn2}
\rput(6.3,2.7){${\hat {\bf F}}_K$}

\nccurve[ArrowInside=->,ArrowInsidePos=0.5,angleA=80,angleB=90]{cn1}{cnK}
\rput(4.5,3.7){${\hat {\bf F}}_1$}
\end{pspicture}
\end{centering}
\caption{CSIT sharing, with (a) and without (b) central node.}
  \label{fig:pic}
\end{figure}
Let us first consider the feedback from a BS to the central node. BS $j$ performs the QR decomposition ${\bf H}_j={\bf F}_j{\bf C}_j$ and quantizes the subspace spanned by the columns of ${\bf F}_j$ using $N_b$ bits and sends the index of the quantized codeword to the central node. We further assume that the BSs and the central node share a predefined codebook ${\mathcal S}=\{{\bf S}_1,...,{\bf S}_{2^{N_b}}\}$ which is composed of $2^{N_b}$ truncated unitary matrices of size $(K-1)N\times M$ and is designed using Grassmannian subspace packing. For simplicity, let us assume that all $K$ codebooks have the same size and the powers of the transmitted signals and receiver noise are symmetric across the network. The quantized codeword is the closest point in $\mathcal S$ w.r.t. the chordal distance, i.e.,  
\begin{equation}\label{edist}
{\hat {\bf F}}_j=\mathrm{arg} \min_{{\bf S} \in {\mathcal S}} \,\,\,d_{c}({\bf S},{\bf F}_j)
\end{equation}
in which $d_c({\bf X},{\bf Y})=\frac{1}{\sqrt{2}}\left|\left| {\bf X}{\bf X}^H - {\bf Y}{\bf Y}^H\right|\right|_{\mathrm{F}}$ is the chordal distance between two points in $ {\mathcal G}_{(K-1)N,M}$ represented by truncated unitary matrices ${\bf X}$ and ${\bf Y} $ \cite{rajathesis}. The interference alignment problem is then solved at the central node based on $\{{\hat {\bf F}}_j\}_{j=1}^K$ to find $(\{{\tilde {\bf U}}_j\}_{j=1}^K,\{{\tilde {\bf V}}_j\}_{j=1}^K)$ fulfilling
\begin{equation}\begin{split}\label{e8}
{\tilde {\bf U}}_{-j}^{\rm H}{\hat {\bf F}}_j{\tilde {\bf V}}_j={\bf 0}, \,\,\,\forall j \in \{1,...,K\}.   % this zero should be a zero matrix
\end{split}\end{equation}
We now consider the feedback of $\tilde {\bf V}_j$ from the central node to BS $j$. Using another codebook ${\mathcal T}=\{{\bf T}_1,...,{\bf T}_{2^{N_c}}\}$ of truncated unitary matrices representing points in $ {\mathcal G}_{M,d}$, the central node quantizes the alignment precoder $\tilde {\bf V}_j$ for each BS on ${\mathcal G}_{M,d}$ according to
\begin{equation}\label{edist2}
{\hat {\bf V}}_j=\mathrm{arg} \min_{{\bf T} \in {\mathcal T}} \,\,\,d_{c}({\bf T},{\tilde {\bf V}}_j),
\end{equation}  
and sends the corresponding index to BS $j$. At BS $j$, % considers ${\hat {\bf F}}_i$ to be equal to ${\bf F}_i{\breve {\bf U}}_i$ (where ${\breve {\bf U}}_i$ is not a unitary matrix necessarily) and therefore the transformation matrix is calculated as 
 we define the total precoder as ${\bf V}_j={\bf C}_j^{-1}{\bf F}_j^{\rm H}{\hat {\bf F}}_j{\hat {\bf V}}_j$, by analogy to the perfect CSI case (Lemma~\ref{lemma_Grassmannian_feedback}). 
%Here we show that our limited feedback scheme preserves the total multiplexing gain of the channel provided that the number of feedback bits grows with the transmit power, and we characterize the required growth rate.
Using the precoders ${\bf V}_j$ and after applying the receive filter ${\tilde {\bf U}}_i$ to \eqref{E1}, the interference leakage (due to the quantizations \eqref{edist} and \eqref{edist2}) at user $i$ is defined as 
\begin{equation}\begin{split}
{\bf e}_i= \sum_{ \substack{1\leq j\leq K\\ j \ne i}}{{\tilde {\bf U}}_i^{\rm H}{{{{\bf H}}}_{ij}}{{{{\bf {\bf V}}}}_j}{{\bf{x}}_j}}. 
\end{split}\end{equation}
We denote the leakage power at user $i$ by $L_i={\rm tr}({\rm E}({\bf e}_i{\bf e}_i^{\rm H}))={\rm tr}(\frac{P}{d}{\bf Q}_I^i)$,
 where ${\bf Q}_I^i=\sum_{j= 1, j\neq i}^K {\tilde {\bf U}}_i^{\rm H}{{{{\bf H}}}_{ij}}{{{{\bf {\bf V}}}}_j}{{{{\bf {\bf V}}}}_j^{\rm H}}{{{{\bf H}}}_{ij}^{\rm H}}{\tilde {\bf U}}_i$.
We now consider the sum over all users of the leakage powers:
\begin{equation}\begin{split}\label{e9}
L&=\sum_{i=1}^{K}{\rm tr}\bigg (\frac{P}{d}\sum\limits_{j= 1, j\neq i}^K {\tilde {\bf U}}_i^{\rm H}{{{{\bf H}}}_{ij}}{{{{\bf {\bf V}}}}_j}{{{{\bf {\bf V}}}}_j^{\rm H}}{{{{\bf H}}}_{ij}^{\rm H}}{\tilde {\bf U}}_i\bigg )\\
%&=\sum_{j=1}^{K}\frac{P}{d}{\rm tr}\bigg (\sum\limits_{i= 1, i\neq j}^K {{{{\bf {\bf V}}}}_j^{\rm H}}{{{{\bf H}}}_{ij}^{\rm H}}{\bf U}_i{\bf U}_i^{\rm H}{{{{\bf H}}}_{ij}}{{{{\bf {\bf V}}}}_j}\bigg )\\
%&=\sum_{j=1}^{K}\frac{P}{d}{\rm tr}\bigg ({\bf V}_j^{\rm H}{\bf H}_j^{\rm H}{\bf U}_{-j}{\bf U}_{-j}^{\rm H}{\bf H}_j{\bf V}_j\bigg )\\
&= \sum_{j=1}^{K}\frac{P}{d}||{\tilde {\bf U}}_{-j}^{\rm H}{\bf H}_j{\bf V}_j ||_{\rm F}^2.\\
\end{split}\end{equation}
Substituting ${\bf V}_j={\bf C}_j^{-1}{\bf F}_j^{\rm H}{\hat {\bf F}}_j{\hat {\bf V}}_j$ and ${\bf H}_j={\bf F}_j{\bf C}_j$ gives
\begin{equation}\begin{split}\label{e99}
||{\tilde {\bf U}}_{-j}^{\rm H}{\bf H}_j{\bf V}_j ||_{\rm F}^2=||{\tilde {\bf U}}_{-j}^{\rm H}{\bf F}_j{\bf F}_j^{\rm H}{\hat {\bf F}}_j{\hat {\bf V}}_j ||_{\rm F}^2.
\end{split}\end{equation}
From \eqref{e8} we have ${\tilde {\bf U}}_{-j}^{\rm H}{\hat {\bf F}}_j{\tilde {\bf V}}_j{\tilde {\bf V}}_j^{\rm H}{\hat {\bf V}}_j={\bf 0}$, therefore by some manipulations, from \eqref{e9}, \eqref{e99} we get
\begin{equation}\begin{split}\label{e799}
%L&=\sum_{j=1}^K \frac{P}{d}||{\tilde {\bf U}}_{-j}^{\rm H}{\bf H}_j{\bf V}_j ||_{\rm F}^2\\
L=\sum_{j=1}^K\frac{P}{d}||{\bf X}_j^b+{\bf X}_j^c||_{\rm F}^2  \leq \sum_{j=1}^K\frac{P}{d}(||{\bf X}_j^b||_{\rm F}+||{\bf X}_j^c||_{\rm F})^2
\end{split}\end{equation}
where
\begin{equation}\begin{split}
{\bf X}_j^b&={\tilde {\bf U}}_{-j}^{\rm H}({ {\bf F}}_j{ {\bf F}}_j^{\rm H}-{\hat {\bf F}}_j{\hat {\bf F}}_j^{\rm H}){\hat {\bf F}}_j{\hat {\bf V}}_j \,\,\,\,\, {\rm and}\\
{\bf X}_j^c&={\tilde {\bf U}}_{-j}^{\rm H}{\hat {\bf F}}_j({\hat {\bf V}}_j{\hat {\bf V}}_j^{\rm H}-{\tilde {\bf V}}_j{\tilde {\bf V}}_j^{\rm H}){\hat {\bf V}}_j \, . 
\end{split}\end{equation}
Using the fact that all the matrices involved in ${\bf X}_j^b$ and ${\bf X}_j^c$ are truncated unitary, it can be shown that $||{\bf X}_j^b||_{\rm F} \leq \sqrt{2d} \, d_c({{\bf F}}_j,{\hat {\bf F}}_j)$ and $||{\bf X}_j^c||_{\rm F} \leq \sqrt{2d}  \, d_c({\tilde {\bf V}}_j,{\hat {\bf V}}_j)$.
Using bounds on the quantization error for codebooks designed by sphere packing, it can be shown \cite{itw} that $L$ in \eqref{e799} is upper bounded by a constant $c_0$ independent of $P$ when
\begin{equation}\label{ina2ta}
N_b=\frac{G_{b}}{2}{\rm log} P \quad \mathrm{and} \quad N_c=\frac{G_{c}}{2}{\rm log} P, 
\end{equation}
in which $G_{b}=2M((K-1)N-M)$ and $G_{c}=2d(M-d)$ are the real dimension of $ {\mathcal G}_{(K-1)N,M}$ and $ {\mathcal G}_{M,d}$ respectively.
Under the conditions \eqref{ina2ta}, it is clear that the leakage power at every receiver would be bounded by a constant since $L_i \leq L$. 

In order to establish the DoF achievable using the proposed CSI quantization scheme, we provide a lower bound for the achievable rate. 
First consider the following lemma:
\begin{lemma}\label{thdof}
For $N_{b}$ and $N_{c}$ according to \eqref{ina2ta} we have,
\begin{equation}\label{23eq}
{\lim_{P \rightarrow \infty} \frac{\log \left| {{{\bf I}_{{d}}} + \frac{P}{d}{\bf Q}_{\rm S}^i} \right|}{\log P}}=d,
\end{equation}
with ${\bf Q}_{\rm S}^i={\bf U}_i^{\rm H}{\bf H}_{ii}{\bf V}_{i}{\bf V}_{i}^{\rm H}{\bf H}_{ii}^{\rm H}{\bf U}_i$, almost surely.
\end{lemma}
\begin{proof} 
Note that the limit in \eqref{23eq} involves codebooks of increasing size since $N_{b}$ and $N_{c}$ increase with $P$. ${\bf Q}_{\rm S}^i$ does not necessarily admit a limit when $P \rightarrow \infty$ due to the fact that ${\bf U}_i$ and ${\bf V}_{i}$ are functions of the codebook. We tackle this problem by resorting to an argument based on the compactness of the solution space, and show that there exists a series of codebooks of increasing size for which ${\bf Q}_{\rm S}^i$ admits a limit and is full rank a.s. The full proof is similar to the proof of Theorem 2 in \cite{prep}, and is omitted due to space constraints.
%(here, the proof builds upon the convergence of a subsequence of $ {\mathcal W}=\{{\bf W}_n\}_{n=1}^{\infty} $ in which ${\bf W}_n=({\bf U}_{1,n},...,{\bf U}_{K,n},{\hat {\bf V}}_{1,n},...,{\hat {\bf V}}_{K,n}) \in {\mathcal G}_{M,d}^{K}\times {\mathcal G}_{N,d}^{K}$ where ${\bf U}_{j,n}$ and ${\hat {\bf V}}_{j,n}$ are the alignment receive filter and the quantized version of the alignment precoder designed at the central node for BS $j$ at a given SNR of $\frac{P}{d}_n$).
\end{proof}

We are now in the position of proving that the proposed method achieves the full IA DoF:
\begin{theorem}\label{thm_DoF}
The proposed quantization scheme, with $N_{b}$ and $N_{c}$ according to \eqref{ina2ta}, achieves the same DoF as IA under perfect CSI.
\end{theorem}
\begin{proof}
Recall that \eqref{ina2ta} ensures that $L_i \leq c_0$. Therefore,
$\lambda_{\rm max} \left(\frac{P}{d} {\bf Q}_I^i\right) \leq {\rm tr}\left(\frac{P}{d} {\bf Q}_I^i\right)=L_i \leq c_0$, which yields
\begin{equation}\label{inkh}
\log \left| {{{\bf I}_{{d}}} + \frac{P}{d}{\bf Q}_{\rm I}^i} \right| \leq d\log \left(1+\lambda_{\rm max} \left(\frac{P}{d} {\bf Q}_I^i\right)\right) \leq d\log (1+c_0).
\end{equation}
Hence, the achievable rate using the designed precoders and receive filters can be lower-bounded as follows, 
 \begin{eqnarray}
R_q^i&=&\log \left| {{{\bf I}_{{d}}} + \frac{P}{d}({\bf Q}_{\rm S}^i+{\bf Q}_{\rm I}^i)} \right|-\log \left| {{{\bf I}_{{d}}} + \frac{P}{d}{\bf Q}_{\rm I}^i} \right|\\
 &\geq& \log \left| {{{\bf I}_{{d}}} + \frac{P}{d}{\bf Q}_{\rm S}^i} \right|-\log \left| {{{\bf I}_{{d}}} + \frac{P}{d}{\bf Q}_{\rm I}^i} \right|  \label{eq_QI_psd}\\
 &\geq& \log \left| {{{\bf I}_{{d}}} + \frac{P}{d}{\bf Q}_{\rm S}^i} \right|-d\log (1+c_0), \label{in23eq}
\end{eqnarray}
where \eqref{eq_QI_psd} follows from the fact that ${\bf Q}_{\rm I}^i$ is positive semi-definite and the second inequality follows from \eqref{inkh}.
Combining \eqref{in23eq} with Lemma~\ref{thdof} brings us to the conclusion that $\lim_{P \rightarrow \infty} \frac{R_q^i}{\log P} \geq d$, i.e. the full DoF is achieved.
\end{proof} 
\mbox{}\\*

\section{Simulation Results and discussion}
\label{sec:simulation}

\subsection{Performance results using RVQ}
In this section, the performance of the proposed scheme is evaluated through numerical simulations. The performance metric is the sum-rate evaluated through Monte-Carlo simulations employing truncated unitary precoders.
%The achievable sum-rate of the MIMO IC using interference alignment precoders under the assumption that the input signals are Gaussian can be written as
%\begin{eqnarray}\label{E11}
%	{R_{\rm sum}} =    \sum\limits_{i= 1}^K \log \left| {{{{\sigma ^2}}}{{\bf I}_{{N}}} + \sum\limits_{j = 1}^K {{{\bf H}_{ij}}{{\bf {\bf W}}_j}{{\bf {\bf W}}_j^ {\rm H}}{\bf H}_{ij}^{\rm H}} } \right|\nonumber   \\  -   \sum_{i=1}^K \log \left| {{{{\sigma ^2}}}{{\bf I}_{{N}}} + \sum\limits_{j=1, j \ne i}^K {{{\bf H}_{ij}}{{\bf {\bf W}}_j}{{\bf {\bf W}}_j^ {\rm H}}{\bf H}_{ij}^{\rm H}} } \right|.
%\end{eqnarray}
A three-user IC is considered where each BS is equipped with $M=5$ antennas while every receiver has $N=3$ antennas and $d=2$ data streams for each user is considered. Entries of the channel matrices are generated according to $\mathcal{CN}(0,1)$ and the performance results are averaged over the channel realizations.
In in Fig.~\ref{fig:sinr}, the quantized CSI feedback method of Section~\ref{sec:LF} (denoted by ``Proposed'') is compared (for scenario I) to the naive method where the interfering channel matrices from the BSs are independently vectorized, normalized and quantized using $N_b$ bits based on the idea of composite Grassmann manifold \cite{rajac} and finally the indices of the quantized vectors are sent to the central node (denoted by Normalized Channel Composite Grassmann Quantization, NC-CGQ). At the central node, in both cases, each precoder is vectorized, normalized and quantized on ${\mathcal G}_{Md,1}$ using $N_c$ bits, and sent to the corresponding BS.
Figure~\ref{fig:sinr} shows the achievable sum-rate versus transmit SNR ($P$) for $(N_b, N_c)=(5,6)$ and $(N_b, N_c)=(10,12)$ bits. A random codebook is used with codebook entries chosen as independent truncated unitary matrices generated from the Haar distribution. For the independent quantization method, random unit norm vectors are used in the codebook construction. Clearly the proposed scheme outperforms the independent quantization method for the same number of bits.
%\begin{eqnarray}\label{E121}
%	{R'_{\rm sum}} =\sum\limits_{i= 1}^K \log \left| {{{{\sigma ^2}}}{{\bf I}_{{d}}} + \sum\limits_{j = 1}^K {{\bf U}_i^{\rm H}{{\bf H}_{ij}}{{\bf {\bf W}}_j}{{\bf {\bf W}}_j^ {\rm H}}{\bf H}_{ij}^{\rm H}{\bf U}_i} } \right|\nonumber  \\ - \sum_{i=1}^K \log \left| {{{{\sigma ^2}}}{{\bf I}_{{d}}} + \sum\limits_{j=1, j \ne i}^K {{\bf U}_i^{\rm H}{{\bf H}_{ij}}{{\bf {\bf W}}_j}{{\bf {\bf W}}_j^ {\rm H}}{\bf H}_{ij}^{\rm H}{\bf U}_i} } \right|.
%\end{eqnarray}
%
%Results are provided in Figure \ref{fig:IA}. 
%The slope of the curves on Figure \ref{fig:IA} at high SNR gives us an indication of the DoF.  
%It is clear from Figure \ref{fig:IA} that  the slope of the rate curve with quantized feedback matches that of perfect CSI when the number of bits is scaled according to \eqref{eqa}. We have used $N_b=[0, \, 7, \, 13, \, 20, \, 26]$ bits and their corresponding powers ($P=2^\frac{2N_b}{N_G}$) to generate the curve which exhibits the achievable DoF. Simulations were performed only up to SNR of $20 \,{\rm dB}$ due to the complexity associated to the exponential size of the codebook.
%
% In our method, each user $i$ only needs to know ${\tilde {\bf U}}_i$ in order to produce its filter ${\bf G}_i^{\rm H}$ and each BS should send this information to its user. Therefore it is clear that there will be a performance loss after employing the filter ${\bf G}_i^{\rm H}$ at user $i$ compared to the achievable sum-rate in \eqref{E11}. 
%Figure \ref{fig:IA} shows the sum-rate comparison when filters ${\bf G}_i^{\rm H}, \,\,\, i=1,2,3$ are used at the users. 
\begin{figure}
  \centering
 \resizebox{9cm}{!}{ \begin{tikzpicture}[scale=1]
    \renewcommand{\axisdefaulttryminticks}{4}
    \pgfplotsset{every major grid/.append style={densely dashed},every mark/.append style={solid}}
    \tikzstyle{every axis y label}+=[yshift=-10pt]
    \tikzstyle{every axis x label}+=[yshift=5pt]
    \pgfplotsset{every axis legend/.append style={cells={anchor=west},fill=white, at={(1.15,1.05)}, anchor=north east,{font=\scriptsize}}}
    \begin{axis}[ 
      %ybar,
      xmin=0,
      ymin=3,
      xmax=30,
      ymax=25.3,
      bar width=3pt,
      grid=major,
      %ymajorgrids=false,
      scaled ticks=true,
      %scale ticks above={4},
      ylabel={Sum-rate [bits/s/channel use]},
      xlabel={{\rm SNR} ($P$) [dB]}
      ]

     \addplot[black,smooth,mark=triangle] plot coordinates{(0,5.7829)    (5,10.6166)    (10,17.2227)   (15,25.2484)          } ;
    
     \addplot[blue,smooth,mark=diamond] plot coordinates{(0,5.2237)    (5,8.7412)    (10,12.3339)    (15,15.1545)    (20,16.8856)   (25,17.7469)   (30,18.1105)   (35,18.2459)   (40,18.2921)};

     \addplot[red,dashed,mark=o,every mark/.append style={solid}] plot coordinates{(0,4.0831)    (5,5.8942)    (10,7.2240)    (15,7.9596)    (20,8.2895)    (25,8.4183)    (30,8.4639)    (35,8.4791)    (40,8.4840)};

    \addplot[blue,smooth,mark=x] plot coordinates{(0,5.0770)    (5,8.1657)    (10,11.0892)    (15,13.2172)    (20,14.4194)    (25,14.9678)    (30,15.1820)   (35,15.2572)    (40,15.2821)};

     \addplot[red,dashed] plot coordinates{(0,3.9475)    (5,5.6988)    (10,6.9805)    (15,7.6772)    (20,7.9768)    (25,8.0869)    (30,8.1240)  (35,8.1359)    (40,8.1397)};

 \legend{ {Perfect CSI}, {$N_b=10, N_c=12$, Proposed},{$N_b=10, N_c=12$, NC-CGQ},{$N_b=5, N_c=6$, Proposed}, {$N_b=5, N_c=6$, NC-CGQ}}
    \end{axis}
  \end{tikzpicture} }
  \caption{Sum-rate comparison of quantization methods, for the $3$-user MIMO IC, $M=5$, $N=3$, $d=2$.}
  \label{fig:sinr}
\end{figure}
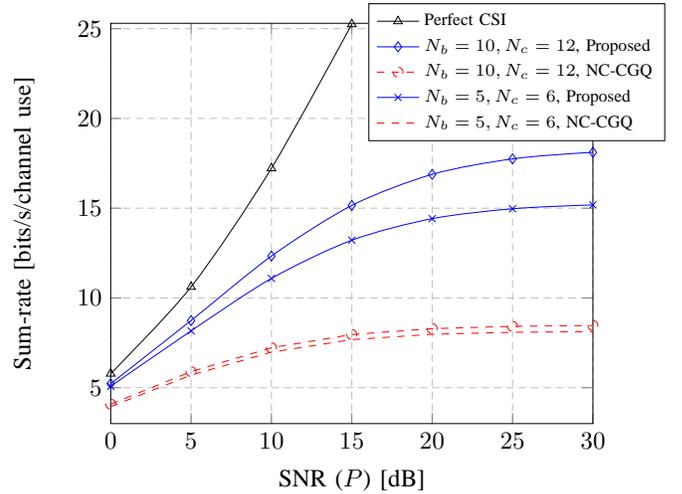
\subsection{Perturbations on the Grassmann manifold}
\label{GMperturbation}
The use of random codebooks for large values of $N_b$ and $N_c$ is not tractable, due to the exponential requirements in terms of storage and of computation of \eqref{edist} and \eqref{edist2}.
In order to benchmark the sum-rate achievable under the proposed scheme for the high power region (large $N_b$ and $N_c$) we replace the quantization process with a perturbation which approximates the quantization error. As will be seen, this approach provides a good approximation of the effect of quantization on the considered system. We now detail the proposed perturbation technique. 

Let us consider a point on ${\mathcal G}_{n,p}$, represented by a $n \times p$ truncated unitary matrix $\bf F$.  Here, we assume that $n\geq2p$, since it is otherwise more efficient to consider the left null space of $\bf F$ instead.
Since the columns of $\bf F$ are orthonormal, they can be completed to form an orthonormal basis of the $n$-dimensional space.
In fact, according to \cite{Barg}, any other point on ${\mathcal G}_{n,p}$ can be represented in the basis constituted by the columns of the unitary matrix ${\bf W} = \left[{\bf F} \,\, {\bf F}^{\rm c}\right]$ as
\begin{equation}\label{eq_Fbar}
{\bf \bar F}={\bf W} \begin{bmatrix} {\bf C} \\ {\bf S} \\ {\bf 0}_{n-2p} \end{bmatrix}\!\!,
\end{equation}
for some ${\bf F}^{\rm c}$ in the left null space of $\bf F$ and  
\begin{equation*}\label{i33}
{\bf C}={\rm diag}(\cos\theta_1,\cdots,\cos\theta_p), {\bf S}={\rm diag}(\sin\theta_1,\cdots,\sin\theta_p), 
\end{equation*}
%\begin{equation}\label{i33}
%{\bf C}=\left [ \begin{array}{ccc} \cos\theta_1 & \cdots & 0 \\ \vdots & \ddots & \vdots \\ 0 & \cdots & \cos\theta_p \end{array}\right ] \!\!, \,\,{\bf S}=\left [ \begin{array}{ccc} \sin\theta_1 & \cdots & 0 \\ \vdots & \ddots & \vdots \\ 0 & \cdots & \sin\theta_p \end{array}\right ] 
%\end{equation}
where $\theta_1,...,\theta_p$ are real angles. 
The squared chordal distance between the two points on ${\mathcal G}_{n,p}$ represented by $\bf F$ and ${\bf \bar F}$ is
%\begin{equation}
$r \triangleq  d_c^2({\bf F},{\bf \bar F})=\sum_{i=1}^p \sin^2\theta_i$.
%\end{equation}
Therefore, in order to generate random perturbations of a certain chordal distance $\sqrt{r_0}$ from $\bf F$, we propose to generate random values for the angles $\theta_1,...,\theta_p$ such that $\sum_{i=1}^p \sin^2\theta_i=r_0$, and to pick a random orthonormal basis ${\bf F}^{\rm c}$ of the left null subspace of ${\bf F}$. The perturbed matrix is then computed using \eqref{eq_Fbar}. 
The histogram (not shown) of the squared quantization error obtained from an RVQ implementation suggests that the Gaussian distribution is a good approximation for the probability density function of $r$. In fact, the moments of this distribution can be obtained using the following result from \cite[Theorem~6]{rajathesis}: for asymptotically large codebook size, when using a random codebook $\mathcal C$ of size $J$ for quantizing a matrix ${\bf F}$ arbitrarily distributed over an arbitrary manifold, the $k$-th moment of the quantization error distribution, $D^{(k)}={\rm E}_{{\mathcal C},{\bf F}}(d_c^k({\hat {\bf F}},{\bf F}))$ for ${\hat {\bf F}}=\mathrm{arg} \min_{{\bf C} \in {\mathcal C}} \,\,\,d_{c}({\bf F},{\bf C})$, can be bounded as
\begin{equation}\label{in27}
\frac{G}{{{(G+k)}}{(c\, J)}^{\frac{k}{G}}} \leq D^{(k)} \leq   \frac{\Gamma ({\frac{k}{G}})}{{\frac{G}{k}}{(c\, J)}^{\frac{k}{G}}},
\end{equation} 
where $c$ and $G$ are respectively the coefficient of the ball volume and the real dimension of the corresponding manifold (here, the Grassmann manifold). 
Note that \eqref{in27} only provides bounds on $D^{(k)}$, however since both the upper and lower bounds are asymptotically tight when the codebook size increases, we arbitrarily choose to use the upper bound as an approximation of $D^{(k)}$, i.e. 
\begin{equation}
{\bar r} \triangleq  \frac{\Gamma ({\frac{2}{G}})}{{\frac{G}{2}}{(c\, J)}^{\frac{2}{G}}} \approx D^{(2)} 
\end{equation}
is the average and 
\begin{equation}
\sigma^2_r \triangleq \frac{\Gamma ({\frac{4}{G}})}{{\frac{G}{4}}{(c\, J)}^{\frac{4}{G}}} - {\bar r}^2 \approx  D^{(4)}-(D^{(2)})^2 
\end{equation}
is the variance. We propose generate the values for $r$ according to $\mathcal{N}({\bar r},\sigma^2_r)$ truncated to $\mathbb{R}^+$. This process is summarized in Algorithm~\ref{algo_perturbation}.
\begin{algorithm}[h!] 
  \caption{Generating random perturbations around ${\bf F}$} \label{algo_perturbation}
     \begin{itemize}
%     	    \item Set the approximate value for average ($\bar r$) and variance ($\sigma^2_r$) of $r=d_c^2({\bf F},{\bf \hat F})$ from the upper bound in \eqref{in27} 	
	     \item Draw a random realization of the squared chordal distance $r$ from $\mathcal{N}(\bar r,\sigma^2_r)$ %a Gaussian distribution with the parameters $\bar r$ and $\sigma^2_r$.
	     \item If $r<0$ then generate a new sample
	     \item Generate independent $s_1,\ldots , s_p$ drawn uniformly from the interval $\left[0,1\right]$
	     \item Compute the angles $\theta_i=\sin^{-1}\left(\frac{s_i\sqrt{r}}{\sqrt{\sum_{i=1}^p s_i^2}}  \right)$
	     \item Generate a random orthonormal basis ${\bf F}^{\rm c}$ of the left null space of ${\bf F}$ and compute ${\bf \bar F}$ according to \eqref{eq_Fbar}.
             %\item Compute ${\bf \bar F}$ according to \eqref{eq_Fbar}.
     \end{itemize}
  \end{algorithm}
  
  This algorithm was used to simulate the effect of quantization taking place at BSs as well as the central node. The sum-rate performance $\sum_{i=1}^{K}R_q^i$ obtained using the perturbation method is plotted against SNR for various codebook sizes in Fig.~\ref{fig:lb}. For the considered antenna configuration, according to \eqref{ina2ta}, the scaling that is sufficient to achieve the perfect DoF is $N_b=5\log P$ and $N_c=6 \log P$. In the simulations, the codebook sizes are chosen as $N_b=5A$ and $N_c=6A$ for integer values of $A$, and the corresponding SNR is computed according to $P=2^A$. The results are also compared to perfect CSIT sharing.
\begin{figure}
  \centering
  \begin{tikzpicture}[scale=1]
    \renewcommand{\axisdefaulttryminticks}{4}
    \pgfplotsset{every major grid/.append style={densely dashed},every mark/.append style={solid}}
    \tikzstyle{every axis y label}+=[yshift=-10pt]
    \tikzstyle{every axis x label}+=[yshift=5pt]
    \pgfplotsset{every axis legend/.append style={cells={anchor=west},fill=white, at={(1,0)}, anchor=south east,{font=\scriptsize}}}
    \begin{axis}[ 
      %ybar,
      xmin=0,
      ymin=-5,
      xmax=30,
      ymax=45,
      bar width=3pt,
      grid=major,
      %ymajorgrids=false,
      scaled ticks=true,
      %scale ticks above={4},
      ylabel={Sum-rate [bits/s/channel use]},
      xlabel={ {\rm SNR ($P$)}  [dB]}
      ]

     %main
     \addplot[black,smooth,mark=diamond] plot coordinates{(0,4.9458)    (5,9.7868)    (10,16.5164)   (15,24.6617)  (20,33.6494)     (25,42.9955)    (30,52.3556)    (35,61.5125)    (40,70.3418)     } ;

     \addplot[blue,mark=triangle] plot coordinates{(0,3.1659)    (6.02,8.0277)    (9.03,11.34)    (12.04,15.48)    (15.05,20.25)  (21.07,30.05)};

     \addplot[blue,mark=square,every mark/.append style={solid}] plot coordinates{(0,4.6509)    (5,8.3443)    (10,12.1937)    (15,15.1608)    (20,16.8912)    (25,17.6821)    (30,17.9833)    (35,18.0863)    (40,18.1197)};

     \addplot[blue,mark=o,every mark/.append style={solid}] plot coordinates{(0,4.4619)    (5,7.5711)    (10,10.4417)    (15,12.3506)    (20,13.2952)    (25,13.6684)    (30,13.7976)    (35,13.8398)    (40,13.8533)};

   \addplot[red,dashed,mark=triangle,every mark/.append style={solid}] plot coordinates{(0,3.1659)    (6.02,8.34)    (9.03,11.5473)    (12.04,15.6829)    (15.05,20.4590)   (21.07,30.2489)   (27.09,42.0428)   (33.11,52.8688)   (39.13,64.2010) };
   
   \addplot[blue,dashed,mark=o,every mark/.append style={solid}] plot coordinates{(0,4.3598)    (5,7.4398)    (10,10.2799)    (15,12.1623)    (20,13.0995)    (25,13.4757)    (30,13.6080)    (35,13.6515)    (40,13.6655)};
  
   \addplot[blue,dashed,mark=square,every mark/.append style={solid}] plot coordinates{(0,4.6263)    (5,8.3045)    (10,12.0895)    (15,14.9204)    (20,16.5011)    (25,17.1939)    (30,17.4507)    (35,17.5374)    (40,17.5654)};

 \legend{{Perfect CSI}, {$N_b, N_c$ according to \eqref{ina2ta}}, {$N_b=15, N_c=18$},{$N_b=10, N_c=12$}}
    \end{axis}
  \end{tikzpicture}
  \caption{Sum rate comparison between perturbation method (dashed) and quantization (solid), for the $3$-user MIMO IC, $M=5$, $N=3$ and $d=2$.}
  \label{fig:lb}
\end{figure}
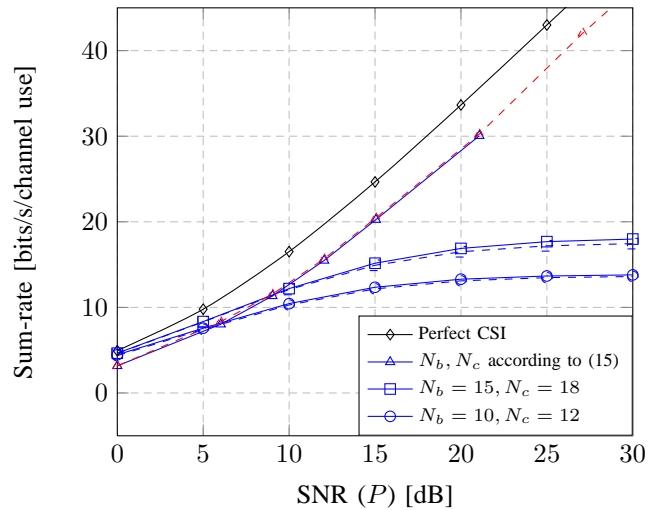
 It is clear that this perturbation method effectively approximates the quantization process when the desired performance metric is the sum-rate, allowing us to rely on the curves resulting from this method to confirm the DoF result of Theorem~\ref{thm_DoF}. 
%(which was not possible at high powers due to exponentially increasing size of RVQ codebooks).
%\begin{figure}[!t]
%\centering
%\includegraphics[width=3.5in]{fig1.eps}
%\caption{Comparison of expected sum-rate }\label{F4}
%\end{figure}
\section{Conclusion}
\label{sec:conclusion}

A limited feedback scheme was proposed for efficient sharing of CSIT among interfering BSs in downlink interference alignment for TDD cellular systems. The growth rate of the bits to be transferred with respect to the transmit power was characterized in order to preserve the total multiplexing gain and a heuristic method was proposed to verify the achievability of multiplexing gain by simulation.
{\small 
\bibliographystyle{IEEEtran} }
%\begin{thebibliography}{1}

  %0.5em minus 0.4em\relax Harlow, England: Addison-Wesley, 1999.
\balance
%\end{thebibliography}
\bibliography{mybib2}

\end{document}